\documentclass[11 pt]{article}
\usepackage{amsthm,amsmath,amsfonts,mathtools,hyperref}
\usepackage{cite,setspace} 
\usepackage{algorithm}
\usepackage[noend]{algorithmic}
\usepackage[a4paper, total={8.5in, 11in},margin=1in]{geometry}

\newtheorem{defin}{Definition}[section]
\newtheorem{theorem}{Theorem}[section]

\newtheorem{lemma}[theorem]{Lemma}

\newtheorem{prob}{Problem}

\newcommand{\D}{{\cal D}}
\newcommand{\G}{{{\cal G}}}
\newcommand{\calF}{{\cal F}}
\newcommand{\E}{\mathbb E}
\newcommand{\cP}{\mathbb P}
\newcommand{\cH}{{\cal H}}
\newcommand{\bigO}{{ O}}
\newcommand{\bigT}{{\Theta}}
\newcommand{\ER}{Erd\"{o}s-R\'enyi }

\newcommand{\quer }{{\textsc{queried}}}
\newcommand{\marked }{{\textsc {marked}}}
\newcommand{\stared }{{\textsc {starred}}}
\newcommand{\obs }{{\textsc {unmasked}}}

\DeclareMathOperator{\tmix}{t_{mix}}
\DeclareMathOperator{\davg}{d_{avg}}
\DeclareMathOperator{\davgs}{d^2_{avg}}
\DeclareMathOperator{\deff}{d_{eff}}

\date{}

\title{On the Complexity of Sampling Vertices Uniformly from a Graph}

\author{Flavio Chierichetti\footnote{Supported in part by the ERC Starting Grant DMAP 680153, by a Google Focused Research Award and by the SIR Grant RBSI14Q743.}\\
Dipartimento di Informatica\\
Sapienza University \\
Rome, Italy \\
{\tt  flavio@di.uniroma1.it}
\and
Shahrzad Haddadan\footnote{Supported in part by the ERC Starting Grant DMAP 680153, by a Google Focused Research Award and by the SIR Grant RBSI14Q743.}\\
Dipartimento di Informatica\\
Sapienza University \\
Rome, Italy \\
{\tt  shahrzad.haddadan@uniroma1.it}
}

\newcommand{\arx}[1]{\href{http://arxiv.org/abs/#1}{\texttt{arXiv:#1}}}

\begin{document}
\maketitle

\begin{abstract}
We study a number of graph exploration problems in the following natural scenario: an algorithm starts  exploring an undirected graph from some {\em seed} vertex; the algorithm, for an arbitrary vertex $v$ that it is aware of, can ask an oracle to return the set of the  neighbors of $v$. (In the case of social networks, a call to this oracle corresponds to downloading the  profile page of user $v$.) The goal of the algorithm is to either learn something (e.g., average degree) about the graph, or to return some random function of the graph (e.g., a uniform-at-random vertex), while accessing/downloading as few vertices of the graph as possible.

Motivated by practical applications, we study the complexities  of a variety of problems in terms of the graph's mixing time $\tmix$ and average degree $\davg$ --- two measures that are believed to be quite small in real-world social networks, and  that have often been used in the applied literature to bound the performance of online exploration algorithms.

Our main result is that the algorithm has to access  $\Omega\left(\tmix \davg \epsilon^{-2} \ln \delta^{-1}\right)$ vertices to obtain, with probability at least $1-\delta$, an $\epsilon$ additive approximation of the average of a bounded function on the vertices of a graph --- this lower bound  matches the performance of an algorithm that was  proposed in the literature.

We also give tight bounds for the problem of returning a close-to-uniform-at-random vertex from the graph. Finally, we give lower bounds for the problems of estimating the average degree of the graph, and the number of vertices of the graph.

\end{abstract}

\section{Introduction}% and the Statement of Results}

Hundreds of millions of people share messages, videos and pictures on Google+ and Facebook each day --- 
 these  media have an increasingly high political, economical, and social importance in today's world. Data miners have consequently devoted significant amounts of attention to the study of large social networks.% and, today, studying these large social graphs is a fundamental problem in data mining.

\smallskip

In data mining, one often seeks algorithms that can return (approximate) properties of online social networks, so to study and analyze them, but without having to download the millions, or billions, of vertices that they are made up of.
The properties of interest range from the order of the graph \cite{Katzir, LowSize}, to its average degree (or its degree distribution) \cite{KumarAvg,AvgFeige, AvgGoldreich}, to the average clustering coefficient \cite{ Clustering, Seshadhri14} or triangle counting \cite{Becchetti:2010}, to non-topological properties such as the average score that the social network's users assign to a movie or a song, or to the fraction of people that {\em like} a specific article or page. All these problems have trivial solutions when the graph (with its non-topological attributes) is stored in main memory, or in the disk: choosing a few independent and uniform at random vertices from the graph, and computing  their contribution  to the (additive) property of interest, is sufficient to estimate the (unknown) value of the graph property --- the empirical average of the contributions of the randomly chosen vertices will be close to the right value with high probability, by the central limit theorem.

\smallskip

In applications, though, it is often impossible to have random access to the (vertices of the) graph. Consider, for instance, an online undirected social graph, such as the Facebook friendship graph. An algorithm can {\em download} a webpage of a given (known) user {\tt alice} from this social graph (e.g., {\tt http://sn.com/user.php?id=alice}), parse the HTML, and get the URLs of the pages of her friends (e.g., {\tt http://sn.com/user.php?id=bob}, {\tt http://sn.com/user.php?id=charles}, etc.) and that user's non-topological attributes (e.g., the set of  movies she likes) --- an algorithm, though, cannot download the webpage of a vertex without  knowing its URL: thus, to download a generic vertex {\tt zoe} from the graph, the algorithm first needs to download all the vertices in (at least) one path from the seed vertex (e.g., {\tt alice}) to {\tt zoe}.

\smallskip

Clearly, given enough many resources, the algorithm could crawl the whole social network (that is, download each of the social network's vertices), and then reduce the problem of computing the online graph property to the centralized one --- unfortunately, it is practically infeasible to download millions, or billions, of vertices from a social network (the APIs that can be used to access the network usually enforce strict limits on how many vertices can be downloaded per day). Several techniques have been proposed in the literature for studying properties of online graphs --- almost all of them assume to have access to a random oracle that returns a random vertex of the graph according to a certain distribution (usually, either uniform, or proportional to the degree), e.g., \cite{Bressan, numedges, Katzir, KumarAvg, AvgGoldreich}.

When running algorithms on online social networks, it is often hard or impossible to implement a uniform-at-random random oracle, and to get samples out of it --- the complexity of this oracle is one of the main problems that we tackle in this paper.

\smallskip

In practice, an algorithm is given (the URL of) a seed vertex (or, some seed vertices) of the social network; the algorithm has to download that seed vertex, get the URLs of its neighbors, and then decide which of them to download next; after having downloaded the second node, the algorithm (might) learn of the existence of some other URLs/vertices, and can then decide which of the known (but unexplored) URLs to download --- and so on, and so forth, until the algorithm can return a good approximation of the property to be guessed.

\smallskip

The natural cost function in this setting is the (random) number of vertices that the algorithm has to download, or {\em query}, before making its guess --- the cost function is usually bounded in terms of properties of the graph (e.g., its order, its average degree, etc.), and in terms of the quality of the algorithm's guess.

\smallskip

Many problems of this form can be found in the literature. In this paper, we consider two natural problems, that are at the heart of many others, and whose complexity (as far as we know) was open before this work:\begin{itemize}
\item the ``average score problem'': assuming that each vertex holds some score in $[0,1]$,  compute an approximation of the average of the scores;
\item the ``uniform at random sample'': return a random vertex from the graph whose distribution is approximately uniform.% (Also second/third, $\zeta$-moment)
\end{itemize}
In a sense, the latter problem is technically more interesting than the former (a solution to the latter provides a solution to the former). In practice, though, the average score problem is much more significant (and ubiquitous), given its many applications \cite{Dasgupta12, Ahn07, Gjoka:2010,Mislove07} (e.g., computing the favorability rating of a candidate, or the average star-rating of a movie).
%to be more interesting.
%These two problems are at the heart of many others, as we will see below. In particular, a solution to the latter problem can be directly used for solving the former.

\smallskip

A number of algorithms have been proposed for  the uniform-at-random sample problem \cite{FlavioKumar,Katzir} --- the best known algorithms require roughly $\tilde{O}\left(\tmix \cdot \davg\right)$ vertex queries/downloads to return a vertex whose distribution is (close to) uniform at random, where $\tmix$ is the mixing time of the lazy random walk on the graph, and $\davg$ is its average degree\footnote{Real-world social networks are known to have a small average degree $\davg$; their mixing time $\tmix$ has been observed \cite{Leskovec09}  to be quite small, as well.}. These algorithms do not use any knowledge of the average degree of the graph $\davg$, but need to know a constant approximation of its mixing time $\tmix$.
To our knowledge, the best  lower bound for the uniform-at-random sample problem before this work,  was $\Omega(\tmix + \davg)$ \cite{FlavioKumar} --- one of the main results of this paper is (i) an almost tight lower bound of $\Omega(\tmix \cdot \davg)$ for this problem, for wide (in fact, polynomial) ranges of the two parameters.\footnote{Observe that the two parameters have to obey some bound for such a lower bound to hold --- in general, any problem can be solved by downloading all the $n$ vertices of the graph: thus, if $\tmix \cdot \davg > \omega(n)$, one can solve the uniform-at-random sample problem with less than $o(\tmix \cdot \davg)$ vertex queries.} The lower bound holds even for algorithms that know constant approximations of $\davg$.

\smallskip

Our lower bound construction for the uniform-at-random sample problem also provides (ii) a tight lower bound of $\Omega\left(\davg \tmix\right)$ for the average score problem --- in fact, we resolve the complexity of the average score problem by showing that our lower bound coincides with the complexity of some previously proposed algorithms, whose analysis we improve.
%We prove that this lower bound is tight
%and provides a tight lower bound for the average score problem, and to a non-trivial lower bound for the problem of estimating the order of the graph.

The same lower bound construction further resolves (iii) the complexity of the average-degree estimation problem, and (iv) entails
 a non-tight, but significant, lower bound for the problem of guessing the graph order (that is, the number of vertices in the graph).

\smallskip

It is interesting to note that all the algorithms that were proposed  require $O(\log n)$ space, while our lower bounds hold for general algorithms with no space restriction. Thus, the problems we consider can be solved optimally using only tiny amounts of space.

%\smallskip
%Finally, we study the role of different advices about the graph that can be given to algorithms. We  show that, if a $1\pm \epsilon$ approximation of $\davg$ (in addition to a constant approximation of $\tmix$) is available,  then there exists an algorithm having complexity $O(\tmix)$.\footnote{FC: we should add the terms that are missing from the bounds.}
%That is, knowing the average degree $\davg$ of the graph (or a $1 \pm \epsilon$ approximation of its average degree)  lowers the complexity of the average score problem by a $\davg$ factor.
%\subsection{Related Works}

\section{Preliminaries}

Consider a connected and undirected graph $\G=\langle V_{\G}, E_{\G}\rangle$ with no self-loops (e.g., the Facebook friendship graph), and a bounded function on its vertices $\calF:V_{\G}\rightarrow [0,1]$.\footnote{Note that from any bounded function we can get a function with range $[0,1]$, through a simple affine transform. Therefore, our results can be trivially extended to functions with any bounded range.}
We aim to estimate the average value of this function, i.e., $f_{avg}=\sum_{v\in V_{\G}}\calF(v)/n$ where $n=\vert V_{\G}\vert$. 
%Interesting examples of this problem are ubiquitous. For instance, the average of number of friends, number of pictures they had posted in 2013, etc.

\smallskip

Motivated by  applications,  we assume that accessing the graph is a costly operation, and that there is little or no information about its global parameters  such as the average degree, the number of vertices or the maximum degree. However,  we can access a ``friendship" oracle: that is, an oracle which, given a vertex $v\in V_\G$, outputs references (their ids, or their URLs) to its neighbors $N_{v}=\{u\in V_\G\vert (v,u)\in E_{\G}\}$.
In such a setting, it is natural to approximate $f_{avg}$ by taking samples from a Markov chain based on the graph structure (see, e.g., \cite{CooperParams, ChernofMark}). A simple random walk on the graph, though, will not serve our purposes since it  samples vertices with probability proportional to their degree, while our goal is to take a uniform average of the values of $\calF$.

On a graph $\G=\langle V_{\G},E_{\G}\rangle$, a lazy simple random walk is a Markov chain which being at vertex $v\in V_{\G}$, stays on $v$ with probability $1/2$ and moves to $u\in N_v$ with probability $1/(2\deg(v))$. Given that $\G$ is connected, the lazy random walk will converge to its unique stationary distribution which we denote  by $\Pi^1$ and which is equal to  $\Pi^1(v)=\deg(v)/2\vert E_{\G}\vert$, $\forall v\in V_{\G}$.

  By $\tmix(\G)$ we refer to the mixing time of the lazy random walk on $\G$, which is the minimum integer satisfying: for any $\tau\geq \tmix(\G)$, $\left|X^{\tau}- \Pi^1 \right|_1\leq 1/4$, where $X^{\tau}$ is the distribution of the lazy walk at time $\tau$, and $\left|\cdot\right|_1$ is the $1$-norm of a vector. Note that by the theory of Markov chains, by taking $\tau\geq \tmix(\G)\log(1/\epsilon)$ we have  $\left| X^{\tau}- \Pi^1 \right|_1\leq \epsilon$.
We denote the uniform distribution on vertices of $\G$ by $\Pi^0$, i.e., $\Pi^0(v)=1/\vert V_{\G}\vert, \forall v \in \G$. In general, we denote a distribution on $V_{\G}$ weighing each vertex $v\in V_{\G}$ proportional to $\deg(v)^{\zeta}$ by $\Pi^{\zeta}$. We may drop all the subscripts when doing so does not cause ambiguity.

Following the framework of \cite{FlavioKumar}, we consider two measures of time complexity. First the number of  downloaded vertices, and second the number of steps the algorithm takes to produce the output. Note that accessing an already downloaded vertex has a negligible cost, and hence, the most relevant cost of the algorithm is the number of downloaded vertices. As mentioned in the introduction,  the algorithms considered in \cite{FlavioKumar} and in this paper, only require space to store constantly many vertices, while our lower bound results hold regardless of the space complexity of the algorithms  %(Theorem \ref{main12} and \ref{main22}). \vspace{2 mm}

\smallskip

\textbf{Our Contribution.} We begin by discussing the problem of producing an approximately-uniform sample vertex from an unknown graph (Problem \ref{pr1}); showing that some algorithm presented in the literature are optimal (Theorem \ref{main21}). Then, we proceed to  the problem of estimating $f_{avg}$ for a bounded function $\calF:V_{\G}\rightarrow[0,1]$ (Problem \ref{pr2}). We extend the positive results of \cite{FlavioKumar}; we particularly study one algorithm, the ``Maximum Degree algorithm'', which we show to be optimal in the number of downloaded vertices. This algorithm requires knowledge of some constant approximation of the graph's mixing time, and and upper bound on its maximum degree --- we also mention in the Appendix three other algorithms, two of which had been proposed in \cite{FlavioKumar}, that give non-optimal bounds for some of the problems we consider. We also show new lower bounds for constant approximations of the order and the average degree of a graph. A summary of our contribution is presented in Table \ref{table1}.

\begin{center}
\begin{table}
\small
\begin{tabular}{|c|c|c|}
\hline
&&\\
&Upper Bound & Lower Bound\\
&&\\
\hline
Average of  a &   $O(\tmix\davg\log(\delta^{-1}) \epsilon^{-2})$&$\Omega(\tmix\davg\log(\delta^{-1}) \epsilon^{-2})$\\
 Bounded Function& {\footnotesize (Theorem~\ref{main12}, with an Algorithm of \cite{FlavioKumar})} &{\footnotesize (Theorem~\ref{main22})} \\
 \hline
Uniform Sample&$O(\tmix \davg\log(\epsilon^{-1}))$ &$\Omega(\tmix \davg)$\\
&{\footnotesize (\cite{FlavioKumar}\;)}&{\footnotesize (Theorem~\ref{main21})} \\
\hline

  Number of Vertices &$O(\tmix\max\{\davg, \vert\Pi^{1}\vert_2^{-1/2}\}\log(\delta^{-1})\log(\epsilon^{-1})\epsilon^{-2})$  &$\Omega(\tmix\davg)$ \\
  &{\footnotesize (\cite{Katzir}\;)}&{\footnotesize (Theorem~\ref{main23})} %\tiny{(tight if } $\sum_{k=1}^{n}k^2\rm{pr}_k>n)$
  \\
  \hline
  Average Degree&$O({\cal D}^2\tmix\davg\log(\delta^{-1})\epsilon^{-2})$&$\Omega(\tmix\davg)$\\
   &{\footnotesize (Application of Theorem \ref{main12})}&{\footnotesize (Theorem~\ref{main23})}\\
\hline
\end{tabular}
\caption{\label{table1}Upper bounds and lower bounds on number of queried vertices for algorithms which explore the graph using a neighborhood oracle and a seed vertex. As mentioned before, $\tmix$ is the mixing time of the lazy random walk on the graph, $\davg$ is its average degree, $\cal D$ is an upper bound on its maximum degree, $\Pi^1$ is its stationary distribution,  and $\epsilon$ and $\delta$ are the precision parameters. The lower bounds for estimating the number of vertices and the average degree hold for any constant approximation.}
\end{table}
\end{center}%As we pointed out in the introduction, knowing the average degree is a  huge advantage and expedites the algorithm. Unfortunately, calculating the average degree in a hard problem itself, many computer scientists have been intrigued by it (\cite{KumarAvg, AvgFeige, AvgGoldreich}), and the best known algorithm, by Dasgupta et al. \cite{KumarAvg} (2014) takes $O(\log({\cal D})\log\log({\cal D}))$ ($\cal D$ being and upper bound on the maximum degree) samples from a weighted distribution on vertices to estimate it. We show in Theorem \ref{main23} that under the assumption of having no information (other than $\tmix$) about the global parameters of the graph, in particular $\cal D$, we need at least  $\tmix\davg$ downloads to obtain a constant approximation of the average degree (Remark \ref{Coro1}). Our result improves Goldreich et al. \cite{AvgGoldreich} result which states an $\epsilon$-approximation requires $\sqrt{{\vert V_{\G}\vert}/\epsilon{\davg}}$ steps. 
In Section \ref{sec2}, we prove our lower bound results  on the number of oracle calls for the following problems: sampling a vertex, learning the order, and the average degree of the graph. Estimations of these parameters in a graph are intertwined meaning with a knowledge about one of them the complexity of estimating the other one changes. For instance, Goldreich and Ron \cite{AvgGoldreich} show that, if  a uniform sample generator is accessible at zero cost (alternatively, if the order of graph is precisely known), then the average degree is computable in $\sqrt{\vert V_{\G}\vert/\davg}$ steps. Our lower bounds for  the aforementioned  problems hold if the algorithm has no $\epsilon-$approximation of the order, and of the average degree of the graph. On the other hand, the lower bound we obtain for  an $\epsilon-$approximation of a bounded function's average holds even if the graph's structure is precisely known. 

\vspace{2 mm}

\noindent
\textbf{Number of downloads to produce a close-to-uniform sample.} We prove a lower bound of $\Omega{(\tmix\davg)}$,  thus, showing that the rejection algorithm and the maximum degree algorithm suggested in the literature \cite{FlavioKumar}   are optimal (Theorem \ref{main21}).

\vspace{2 mm}

\noindent
\textbf{Number of downloads to estimate the number of vertices.} The problem of estimating the order of a graph is widely  studied \cite{CooperParams, Katzir}. Katzir et al. \cite{Katzir} (2011) propose an algorithm that, having access to an oracle that produces random vertices from the graph's stationary distribution, requires $\max\{\frac{1}{\vert \Pi^1\vert_{2}}, \davg \}(\frac{1}{\epsilon^2\delta})$ samples to obtain  an $\epsilon$ approximation with probability at least $1-\delta$. It has been shown the number of samples in Katzir's algorithm  is necessary (\cite{LowSize}).
The Katzir et al. algorithm implies an upper bound of $\tmix\max\{\frac{1}{\vert \Pi^1\vert_{2}}, \davg \}(\frac{\log(\epsilon^{-1})}{\epsilon^2\delta})$ vertex queries to obtain  an $\epsilon$ approximation with probability at least $1-\delta$ in our friendship-oracle model. 
In Theorem \ref{main23} we present a lower bound on the number of accesses to the vertices, to get a constant approximation of the graph's order in our friendship-oracle model. %Our lower bound does not make any assumption about the existence of a sampling oracle.
Our lower bound is tight for the graphs that satisfy $\frac{1}{\left| \Pi^1\right|_{2}}< \davg$ --- that is, the graphs whose variance of the degree distribution is greater than $n$.\footnote{
Let $\text{pr}_{k}$ be the fraction of vertices with degree $k$. We have $\left| \Pi^1\right|_{2}=n\sum_{i=1}^{n} \text{pr}_{k}\frac{k^2}{4\vert E\vert^2}= \frac{1}{n\davg^2}\sum_{i=1}^{n} \text{pr}_{k}k^2 $. Thus, to have $1/\sqrt{\left| \Pi^1\right|_{2}}\leq \davg$, it is necessary and sufficient to have $\sum_{i=1}^{n}k^2 \text{pr}_{k}>n$.
}
This class include, say, all the  graphs having a power-law degree distribution with exponent smaller than $3/2$ (e.g., social networks \cite{LKF05}).

\vspace{2 mm}

\noindent
\textbf{Number of downloads to estimate the average degree.} There are quite a few results on estimating the average degree of a graph. The first one by Feige et al. \cite{AvgFeige} introduced a sublinear algorithm of complexity $\sqrt{\vert V_{\G}\vert}$ for a 2-approximation. Goldreich et al. \cite{AvgGoldreich} extends Feige et al.'s result and presents an  $(1\pm\epsilon)$ approximation algorithm with running time $\bigO(1/\epsilon)\sqrt{\vert V_{G}/\davg\vert}$ --- they also prove a lower bound on the number of samples of $\sqrt{\vert V_{G}\vert/\davg}$ --- both of \cite{AvgFeige} and \cite{AvgGoldreich} assume to have access to an oracle capable of producing a uniform at random vertex. Recently, Dasgupta et al.  \cite{KumarAvg} showed that by sampling $\bigO (\log({\cal D})\log\log({\cal D}))$ vertices of a graph from some weighted distribution\footnote{Dasgupta et al.  use an oracle samples each $v\in V_{\cal G}$ proportional to $\deg(v)+c$ for some constant $c$. Note that for $c=0$ this distribution will be the same as the stationarity.}
 one can obtain a $(1\pm\epsilon)$ approximation of its average degree,   where $\cal D$ is an upper bound on the maximum degree. By factoring in the the cost of sampling, the complexity becomes $\bigO(\tmix\log({\cal D})\log\log({\cal D}))$. Taking ${\cal D}=n$ and adding the cost of estimating the graph size, takes the upper bound to: $\bigO\left(\tmix\left((\log({n})\log\log({n}))+\davg+\frac{1}{\vert \Pi^1\vert_{2}}\right)\right)$. 

In Theorem \ref{main23} of this paper we show that by downloading $o(\tmix \davg)$ vertices, it is impossible for an algorithm to have any constant approximation of the average degree $\davg$ with probability more than some constant.

\vspace{2 mm}

Finally, our main result is the following lower bound --- unlike the above three lower bounds, this one holds even if we know exactly the graph's structure. 
\vspace{2 mm}

\noindent
\textbf{Number of downloads to find an $\epsilon, \delta$ approximation for the average of a bounded function.} In Theorem \ref{main22}, we show that an algorithm requires $\Omega\left(\tmix\davg(1/\epsilon^2)\log(1/\delta)\right)$ vertex downloads to produce an $\epsilon-$additive approximation of $f_{avg}$, with probability at least $1-\delta$. This lower bound, together with Theorem \ref{main12}, allows us to conclude that the ``maximum degree algorithm'' is an optimal algorithm for this problem. Note that this algorithm has to have some upper bound $\cal D$ on the maximum degree of the graph. In many situations, one can assume that this information is available --- for instance ${\cal D} \le n$ and, in many cases, one can assume to have a constant approximation to the order of the graph (for instance, in Facebook, one could claim that $\cal D$ is no larger than the world's population.)
Observe that the maximum degree algorithm suffers no loss in getting a large $\cal D$, as opposed to a tighter one, since $\cal D$ does not impact the upper bound on the number of downloaded vertices. 

\newpage
\subsection{Statement of Problems and Results}

\noindent
\begin{prob}\label{pr1}
Input: A seed vertex  $s\in V$ in graph  $\G=\langle V_{\G},E_{\G}\rangle$. Output: A random vertex $v \in V_{\G}$ whose distribution is at  total variation distance at most $\epsilon$ from the uniform one on $V_{\G}$.
 \end{prob}

Several algorithms have been proposed for Problem~\ref{pr1} \cite{FlavioKumar,Katzir} --- we will specifically consider the ``maximum degree sampling'' algorithm, the ``rejection sampling'' algorithm, and the ``Metropolis Hasting'' algorithm.

%SHAHRZAD ADDED: The maximum degree and rejection sampling algorithm download $\tilde{O}(\tmix\davg)$ vertices which we show is optimal.
%In the main body of this paper we will mostly deal with the maximum degree algorithm  (Algorithm~\ref{alg:max}), which we will show to be optimal for Problem~\ref{pr1}, as well as another problem (Problem~\ref{pr2}) --- we discuss the other two algorithms in Appendix~\ref{sec:otheralgos}.

\smallskip

The efficiency of the three algorithms has been studied in terms of the number of their running time (or the number of steps they make on the Markov chain they are based on) and, more importantly, on the number of \emph{queries}\footnote{A vertex is ``queried'', when the set of its neighbors is obtained from the oracle for the first time --- equivalently, when it is downloaded.} (or downloaded vertices) that the algorithm performs. The rejection sampling and maximum degree  algorithms  produce a close-to-uniform random vertex by  querying\footnote{To get $\epsilon$ close to the uniform distribution we need $\bigO\left(\tmix \davg\log (\epsilon^{-1})\right)$ downloads.% As mentioned before, we usually take $\epsilon=1/4$.
}   $\tilde{\bigO}\left(\tmix \davg\right)$ distinct vertices from the graph, where $\tmix$ is the mixing time of a simple random walk on $\G$, and $\davg$ is the average degree of $\G$. In terms of space complexity, each of these algorithms is based on a simple random walk on $\G$ and thus only require  space to save constant number of vertices. 

One of the main results of this paper is Theorem \ref{main21}, which shows the optimality of the maximum degree, and of the rejection sampling, algorithms for Problem~\ref{pr1} --- their running time. We observe that our lower bound holds regardless of the amount of space available to the algorithm.

\begin{theorem}\label{main21} % change this to something like thm:lb:sample
For any large enough constant $c$, and for any arbitrary $n$, $d=\omega(\log n)$, and $ t= o(\frac{n}{d^2})$ there exists a distribution over graphs $\G=\langle V,E\rangle$, each having mixing time $\Theta(t)$, $\E(\vert V\vert)=\Theta(n)$,   $\davg=\E_{v\in V}(\deg_v)=\bigT(d)$, such that any algorithm $\cal A$ that queries less than $\davg\tmix/c$ vertices of $\G$, and that returns some random vertex of $\G$ with distribution $\Pi_{\cal A}$, is such that\footnote{Observe that the expected $\ell_1$ distance between the distributions is over the random variable $\Pi_{\cal A}$.}
 $\E\left[\left|\Pi_{\cal A}-\Pi^0\right|_1\right]\geq \frac{24}{100}-\frac{202}{c-1}$.
 
The same lower bounds also hold if one aims to obtain the generic $\Pi^\zeta$ distribution: if $\zeta >1$,  and $d$ and $t$ satisfy $d =o(t^{\frac{\zeta-1}{\zeta}})$, then any algorithm $\cal A$ that queries less than $\davg\tmix/c$ vertices of $\G$, and that returns some random vertex of $\G$ with distribution $\Pi_{\cal A}$, is such that: $\E\left[\left|\Pi_{\cal A}-\Pi^{\zeta}\right|_1\right]\geq \frac{24}{100}-\frac{202}{c-1}$.
\end{theorem} 

\iffalse
The proof of the following Theorem, showing the optimality of the Maximum Degree algorithm, can be found in Appendix~\ref{app:prmain21}.
\begin{theorem}\label{thm:ub:sample}
The Maximum Degree Algorithm returns a UAR sample with ...
\end{theorem}\fi

The above theorem, and the other lower bound results that we  mention in this section, will be proved in Section \ref{sec2}. %When we talk about ``querying'' a vertex, we mean querying or visiting it for the first time, or downloading it. 

 \medskip
 
Then,  we consider the problem of finding the average of a function $\calF$ defined on vertices of a graph and ranging in $[0,1]$.
  \begin{prob}\label{pr2}
\noindent
Input: A seed vertex  $s\in V$ in graph  $\G=\langle V,E\rangle$ --- each vertex $v$ holds a value $0 \le \calF(v) \le 1$ which we  learn upon visiting it. Output: $\bar{f}$ such that
$\cP \left( \vert\bar{f}-f_{avg}\vert\leq  \epsilon \right) \geq 1- \delta$.
 \end{prob}
 
% We consider algorithms which other than the neighbourhood oracle, have access to  $\calF(v)$ if $v\in V$ has been queried. 

Note that having a uniform sampler (the maximum degree or rejection sampling algorithm of \cite{FlavioKumar}), we can have an $\epsilon$ approximation of $f_{avg}$ with probability $1-\delta$ by taking $\bigO(\epsilon^{-2}\log(\delta^{-1}))$ independent samples which are $\epsilon$ close to uniformity. In total, the number of queries will be $\bigO(\tmix\davg\log(\delta^{-1})\epsilon^{-2}\log{(\epsilon^{-1})})$.
Here we propose a slight variation of the  ``maximum degree'' algorithm to obtain a tight upper bound. We improve the analysis of the ``maximum degree algorithm'' in Theorem~\ref{main12} --- its performance beats the other natural three algorithms, and the main result of this paper is that this performance is  optimal (Theorem \ref{main22}).  We  discuss the performances of the other algorithms in Appendix~\ref{sec:otheralgos}.

\begin{algorithm}
\small
\begin{algorithmic}[1]
\REQUIRE Seed vertex $s\in V_{\G}$, a constant approximation of $\tmix$, and an upper bound $\D$ on $d_{max}$
\ENSURE An $\epsilon$ additive approximation of $f_{avg}$ with probability at least $1-\delta$
\STATE{Consider the maximum degree Markov chain: at vertex $v\in V$  go to the generic $u\in N_v$ with probability $1/\D$, otherwise stay at $v$. 

}
\STATE{Starting from $s$, run the chain for $\frac{\tmix}{d_{min}}\D \cdot \left(1+   \epsilon^{-2} \ln \delta^{-1}\right)$ steps --- let $v_0 = s, v_1, v_2, \ldots$ be the states that are visited by the walk}
\STATE{$T\leftarrow\tmix \D/d_{min}$}
%\STATE{$v_i\leftarrow{\cal M}^T$}
\STATE{$S\leftarrow0,t\leftarrow0,i \leftarrow 0$}
\WHILE{$t<\frac{\davg}{d_{min}} (\tmix/\epsilon^2)\log (1/\delta)$}
%\STATE{$v_i\leftarrow{\cal M}^{T+i}$}
\STATE{$i \leftarrow i + 1$}
\IF{$v_i\neq v_{i-1}$}
\STATE{$t \leftarrow t + 1$}
\ENDIF
\STATE{$S=S+ {\cal F}(v_i)$}
\ENDWHILE
\RETURN{$S/i$}
\end{algorithmic}
\caption{\label{alg:max}The Maximum Degree Algorithm.}
\end{algorithm}

\medskip

The proof of the following Theorem can be found in Appendix~\ref{app:mdub}.
\begin{theorem}\label{main12}
Consider a graph $\G=\langle V,E\rangle$, and a function $\calF: V\rightarrow [0,1]$. Let $\tmix$ be the mixing time of the simple lazy random walk on $\G$. Let  $\bar{f}$ be the value returned by 
Algorithm~\ref{alg:max}. Then,

\begin{equation}
\cP\left(\vert \bar{f} -f_{avg} \vert\geq \epsilon\right) \leq  \delta
\end{equation}

This algorithm  queries  $O(\tmix \davg\epsilon^{-2} \log(\delta^{-1}))$ vertices from the graph, and requires space for saving a constant number of them. 
The number of  computational steps it performs is $O({\cal D}\tmix \/\epsilon^{-2} \log(\delta^{-1}))$.

\end{theorem}

The main result of this paper is the following lower bound which complements the upper bound given in the previous theorem:
\begin{theorem}\label{main22}
For any arbitrary $n$, $d=\omega(\log n)$, and $ t= o(\frac{n}{d^2})$ there exists a distribution over graphs $\G=\langle V,E\rangle$ with mixing time $\Theta(t)$, $\E(\vert V\vert)=4n$,   $\davg=\E_{v\in V}(\deg_v)=\bigT(d)$, and a function $\calF: V\rightarrow \{0,1\}$ such that,
any  algorithm $\cal A$ as described above which aims to return the average of $\calF$, with $\epsilon$
precision  for arbitrary $0<\epsilon,\delta <1$, and queries less than $\Omega(\tmix\davg\epsilon^{-2}\log(\delta^{-1}))$  vertices of $\G$ fails with probability greater than $\delta$.
 \end{theorem}

Finally,  we consider the problems of obtaining an approximation of the average degree, and the number of vertices, of a graph:
  \begin{prob}\label{pr3}
\noindent
Input: A seed vertex  $s\in V$ in graph  $\G=\langle V,E\rangle$. Output: an integer $\bar{n}$ such that
$\cP \left( \vert\bar{n}-|V|\vert\leq  \epsilon \right) \geq 1-  \delta$.
 \end{prob}
 By a result of Katzir \cite{Katzir}, we know by taking $\max\{\davg, 1/\sqrt{\vert\Pi^2\vert_2}\}\epsilon^{-2}\delta^{-1}$ samples from the stationary distribution we are capable to  obtain an $\epsilon$ approximation with probability at least $1-\delta$. To implement a sampling oracle using our neighborhood oracle,  we can  run a Markov chain for $\tmix\log (\epsilon^{-1})$ steps. Thus, the runtime will be $\tmix\max\{\davg, 1/\sqrt{\vert\Pi^1\vert_2}\}\log (\epsilon^{-1})\epsilon^{-2}\delta^{-1}$, which for constant $\epsilon$ and $\delta$ is  $\tmix\max\{\davg, 1/\sqrt{\vert\Pi^1\vert_2}\}$. Theorem \ref{main23} provides a lower bound for a constant approximation which is as mentioned before tight when the variance of the degree distribution is greater than $n$.

  \begin{prob}\label{pr4}
\noindent
Input: A seed vertex  $s\in V$ in graph  $\G=\langle V,E\rangle$. Output: an integer $\bar{d}$ such that
$\cP \left( \vert\bar{d}-\davg\vert\leq  \epsilon \right) \geq  1-\delta$.
 \end{prob}
Normalize  the function $deg:V_{\G}\rightarrow{\mathbb R}$ by dividing its value to $\D$. By Theorem~\ref{main12}, Algorithm~\ref{alg:max} provides an $\epsilon$ approximation with probability at least $1-\delta$ after downloading $O({\cal D}^2\tmix\davg\epsilon^{-2}\log(\delta^{-1}))$ many vertices --- that is,  $O({\cal D}^2\tmix\davg)$ many vertices for constant $\epsilon$ and $\delta$.
With Theorem \ref{main23}, we provide a lower bound for a constant approximation.  
 
 \begin{theorem}\label{main23} 
For any arbitrary $n$, $d=\omega(\log n)$, and $ t= o(\frac{n}{d^2})$ there exists a distribution of graphs $\G=\langle V,E\rangle$ with mixing time $\Theta(t)$, $\E(\vert V\vert)=\Theta(n)$,   $\davg=\E_{v\in V}(\deg_v)=\bigT(d)$ such that, for arbitrary constants $c' > 1$ and large enough $c$, any algorithm that queries at most $\davg\tmix / c$ vertices of the graph, and that outputs an estimation  $\bar{n}$ of $n=\vert V\vert$ (resp., an estimation $\bar{d}$ of $\davg=\vert E\vert/n$), has to satisfy $\max\{\bar{n}/n,n/\bar{n}\} >c'$ (resp., $\max\{\bar{d}/\davg,\davg/\bar{d}\}>c'$),  with probability at least $\frac{99}{100}-\frac{202}{c-1}$.  \end{theorem}

\section{Proofs of the Main Theorems}\label{sec2}
 The proof of our Lower Bounds will be based on the following high-level strategy. Nature will first randomly sample a graph $\mathcal{H}$ according to some distribution; then with probability $1/2$, $\mathcal{H}$ will be the unknown graph traversed by the algorithm; with the remaining probability, the algorithm will traverse a graph $\mathcal{G}$ which is obtained from $\mathcal{H}$ by means of a transformation that we call the \emph{decoration construction}. We will prove that, for the right choice of the distribution over $\mathcal{H}$, an algorithm that performs too few queries to the unknown graph will be unable to tell with probability more than $1/2 + o(1)$, whether the unknown graph it is traversing is distributed like $\mathcal{H}$, or like $\mathcal{G}$.
 
 The decoration construction will guarantee that the properties (e.g., number of nodes, average degree, or even the values assigned by the bounded function to the vertices) will be quite far from each other in $\mathcal{H}$ and $\mathcal{G}$. This will make it impossible for the algorithm to get good approximation of any of those properties --- we will also show it impossible for the algorithm to return a close to uniform-at-random vertex (essentially because the decoration construction will add a linear number of nodes to $\mathcal{H}$, and the algorithm will be unable to visit any of them with the given budget of queries.)
 
 \smallskip
 
%In our proofs, we will be assuming for simplicity of exposition that the generic algorithm (for which we aim to prove the lower bounds) always exhausts its budget of queries. ?
 
 \smallskip

We present a roadmap of our proof strategy here: We start by describing the  \emph{decoration construction} which, given any graph $\cH$, produces a graph $\G$ with similar mixing time and average degree, but with a linear number of ``hidden'' new vertices. After presenting the definition for  the decoration construction,  in Definition \ref{defingraph} we introduce a class of graphs to which we apply this construction.  These graphs' mixing time and average degree can be set arbitrarily. Later, in Lemma \ref{finalLemma} we prove that if an algorithm, equipped only with the neighborhood oracle traverses a graph of this type and queries few vertices of it, it will not be capable of finding any of its hidden vertices.  This is our main lemma from which Theorems  \ref{main21},  \ref{main22}, and \ref{main23} can be concluded.  We now proceed to the formal definitions:
 
 \vspace{0.3 cm}
 \begin{defin}
 \textbf{The Decoration Construction.}  Let  $\cH= \langle V,E\rangle$ be an arbitrary graph. We construct  $\G$ from $\cH$ in the following way:

\smallskip

  Take $t:= \tmix(\cH) $, and   mark any vertex $v\in V$ with probability $1/t$. For any marked vertex $v \in V$, add a vertex $v^*$ and connect it to $v$ via an edge. For a constant $c_1$, attach $c_1t-1$ new degree one vertices to $v^*$ --- this makes the degree of $v^*$ equal to $c_1t$. Let this new graph be $\G$.  We denote the set of marked vertices by $\marked$ and the set $\G\setminus \cH$ by $\stared$. By saying a vertex $v$ is $\stared$ ($\marked$) we mean $v\in \stared$ ($v\in \marked$), and to indicate their numbers we use a preceding $\#$.  We call the $\stared$ vertices with degree $c_1t$ the $\stared$ centers.  Note that to any $\stared$ vertex we can associate a unique $\marked$ vertex.
  \end{defin}
In Lemma \ref{mixingG} we will show that the above construction does not change the  mixing time of $\cH$ drastically, i.e., $\tmix (\cH)\leq \tmix(\G)\leq c \tmix (\cH)$, for some constant $c$.

 \begin{lemma}\label{mixingG}
 Take $\cH$ an arbitrary graph with average degree $d=\omega(1)$, $\tmix(\cH) = o(|V(\cH)|)$ and 
 let $\G$ be obtained from $\cH$ after the decoration construction explained above. Then, with high probability over the random construction of $\G$, the mixing time of the lazy simple random walk on $\G$ satisfies:
 \begin{equation}
  c' \tmix (\cH)\leq \tmix(\G)\leq c \tmix (\cH),\text{ for some constant }c \text{ depending on } c_1.
 \end{equation}
 \end{lemma}
To prove Lemma~\ref{mixingG}, we employ the following result by Peres and Sousi \cite{HitMix}:
 \begin{theorem}[Peres and Sousi\cite{HitMix}]\label{SP}
Let $\alpha < 1/2$. Then, there exist positive constants $c_{\alpha}$ and $c'_{\alpha}$ so that for the lazy random walk on any graph $\cal M$:
$ c'_{\alpha} t^{H}_{\cal M}(\alpha)\leq \tmix({\cal M})\leq  c_{\alpha}   t^{H}_{\cal M}(\alpha),$
where $ t^{H}_{\cal M}(\alpha):= \max_{
\substack{x\in {\cal M}, S\subseteq {\cal M};\\ \pi(S)\geq \alpha}
}\E(\tau_x(S))$, and $\tau_x(S)$ is the time to hit $S$ from $x$.
 \end{theorem}
 
 \begin{proof}[Proof of Lemma~\ref{mixingG}]
 By Theorem \ref{SP} the first inequality holds. We prove the second inequality, 
by applying Theorem \ref{SP} to a lazy random walk on $\G$. Take $\alpha=1/4$, and let   $n$ be the order of $\G$.
 Note that  the $\stared$ vertices have average degree $2-O(1/t) \le (2c_1t-1)/c_1t\le 2$. By applying the Chernoff bound on the number of $\marked$ vertices in $\G$ we have:
$$ \cP (\# \marked >2n/t)\leq e^{-n/t}.$$
Thus, with high probability $\# \stared \leq 2c_1n$ and $\Pi^1($\stared$)\leq 4c_1/(d+4c_1)$, and for $d>44c_1$  any set $S$ satisfying $\pi(S)\geq 1/4$ contains vertices in $\cH$ such that $\Pi^1(\cH\cap S)\geq 1/6$. 
Moreover, for any arbitrary $S$  we have,
 $\max^{\G}_{v\in \G}\E(\tau_v (S))=\max \left\{\max_{v\in \cH}\E^{\G}(\tau_v (S))\right.$, $\left. \max_{v\in \G\setminus \cH}\E^{\G}(\tau_v (S))\right\}.$\footnote{
When talking about $\E$ or $\cP$ if the graph that we are referring to is not clear, we use a superscript for $\E$ as well as  for $\cP$ to denote the underlying graph. Thus, $\E^{\G}(\tau_v(S))$ means the expected time to get from $v$ to $S$ in $\G$. If we are taking the expected value of a random variable $X$ over a set $S$ we use subscript, i.e., $\E_{S}(X)= \sum_{s\in S}X(s) \cP{(S)}$, and $\E^{\G}_{S}(X)= \sum_{s\in S}X(s) \cP^{\G}{(S)}$.
 }
 
 We first assume $v\in \cH$, and bound $\E^{\G}(\tau_v (S))$ when $ \Pi^1(S)\geq 1/6$. Employing Theorem \ref{SP}, we get $t^{H}_{\cH}(1/6)\leq (1/c'_{1/6})  \tmix(\cH)$. 
 We want to compare the length of the paths from $v$ to $S$ in $\G$ and $\cH$.
 We denote the set of all such paths in $\G$ by $\Gamma_{v\rightarrow S}(\G)$ and in $\cH$  by $\Gamma_{v\rightarrow S}(\cH)$. For any $l'\in \Gamma_{v\rightarrow S}(\G)$, we will get a path from $\Gamma_{v\rightarrow S}(\cH)$  by removing its $\stared$ vertices. Consider an arbitrary $l\in \Gamma_{v\rightarrow S}(\cH)$, and all the paths in $\Gamma_{v\rightarrow S}(\G)$ which can be associated to $l$ by removal of $\stared$ vertices as mentioned above. We call these paths $l$'s extensions, and we denote the set they constitute by ${\cal E}x(l)$.

Consider an arbitrary $v\in \cH$,  $S\subseteq \cH; \Pi^1(S)\geq 1/6$ and $l\in \Gamma_{v\rightarrow S}(\cH)$.  Let the $\marked$ vertices on $l$ be $v_1, v_2 , \dots , v_{x_l}$, and $S_1,S_2,\dots S_{x_l}$ be the $\stared$ vertices  that are connected to $l$ respectively through $v_1,\dots ,v_{x_l}$. We know that $\E_{l\in \Gamma_{v\rightarrow S}({\cal H}); \vert l\vert=\rho}(x_l)=\rho /t$, and using linearity of expectation we will have $\E_{l'\in {\cal E}x(l) }(\vert l'\vert)= l+ \sum_{i=1}^{x_l}\E_{l'\in {\cal E}x(l) }(\vert l'|_{S_i}\vert)$, where $l'|_{S_i}$ is the part of $l'$ that lies in  $S_i$. Employing Wald's equation we will have:
 $\sum_{i=1}^{x_l}\E_{l'\in {\cal E}x(l) }(\vert l'|_{S_i}\vert)$ $= \E(x_l)\E_{l'\in {\cal E}x(l) }(\vert l'|_{S_1}\vert)$.
Thus:

\begin{equation}
\begin{array}{lllll}
\E^{\G}(\tau_v(S))&=\sum_{l'\in \Gamma_{v\rightarrow S}(\G)}(\vert l' \vert)\cP^{\G}(l')\\
\smallskip
&=\sum_{l\in \Gamma_{v\rightarrow S}(\cH)}\sum_{l'\in {\cal E}x(l)}(\vert l' \vert)\cP^{\G}(l'\vert l'\in {\cal E}x(l) )\cP^{\cH}(l)\\
\smallskip

&=\sum_{l\in \Gamma_{v\rightarrow S}(\cH)}\left(\vert l\vert+\sum_{i=1}^{x_l}\E_{l'\in {\cal E}x(l) }(\vert l'|_{S_i}\vert)\right)\cP^{\cH}(l)\\
\smallskip

&\leq\sum_{\rho=1}^{\infty}
\sum_{ \substack{{l\in \Gamma_{v\rightarrow S}(\cH)}\\\vert l\vert=\rho} }
\left(\vert l\vert+ \E(x_l) 4c_1t\right) \cP^{\cH}( l)\\
&& \text{\footnotesize The time to hit $\cal H$ from a $\stared$ vertex }\\

 &&\text{\footnotesize   is a geometric random variable with  }\\
 &&\text{\footnotesize probability $1/2c_1t$ of success, each step} \\
&&\text{\footnotesize taking $2$ time units}\\ 
&=\sum_{\rho=1}^{\infty}
\left(\rho + \frac{\rho}{t}4c_1t)\right)\cP^{\cH}(\vert l\vert=\rho) \\
\smallskip
&\leq \sum_{\rho=1}^{\infty}
(5c_1\rho)\cP^{\cH}(\vert l\vert=\rho)= 5c_1~ \E^{\cH}(\tau_v (S)).
\end{array}
\end{equation}

 \smallskip

  We now assume  
 $v\in \G\setminus \cH$. Let $w$ be the $\marked$ vertex in $\cH$ closest to $v$. By linearity of expectation we have: 
  $\E(\tau_v(S))= \E(\tau_v(w))+  \E(\tau_w(S))$. The time to hit $w$ is a geometric random variable with probability $1/2c_1t$ of success and each of whose steps take $2$ time units, thus $\E(\tau_v(w))=4c_1t$. Following the same reasoning as given in the previous paragraph we know $\E(\tau_w(S))\leq (5c_1+4c_1)t$.

 \vspace{0.3 cm}
  
   Putting all the above together:
  $$\begin{array}{ll}
  \max_{v\in \G}\E(\tau_{v}(S))=\max \left\{\max_{v\in \cH}\E(\tau_{v}(S)) , \max_{v\in \G\setminus \cH}\E(\tau_{v}(S))\right\}\\ \hspace{6cm}\leq \max \left \{5t c_1/c'_{1/6}, 9tc_1/c'_{1/6}\right\}\\
  \hspace{6cm}= c t.
  \end{array}
  $$
  
  Another application of  Theorem \ref{SP} will yield the result. 
\end{proof}
We now introduce the random graph to which we will apply the decoration construction, and which will be at the heart of our lower bounds. 
\begin{defin} \label{defingraph}
We define the graph $\cH_{n,d,\psi}$ as follows: given arbitrary parameters $n$, $d$, and $0<\psi<1$, take two \ER graphs $H_1=\langle V_{H_1},E \rangle$ and $H_2=\langle V_{H_2},E \rangle$ with parameters $\langle n,d/n \rangle$.  
%For $d=o(\log^2 (n))$, take  two \ER graphs $H_1=\langle V_{H_1},E \rangle$ and $H_2=\langle V_{H_2},E \rangle$ with parameters $\langle m,d/m \rangle$, such that their largest connected component has size $n$.
Choose $\psi n$ vertices uniformly at random from $V_{H_1}$ namely $v_1, v_2,\dots v_{\psi n}$, and then $\psi n$ vertices uniformly at random from $V_{H_2}$ namely $u_1, u_2,\dots n_{\psi n}$. Select a uniformly random permutation $\sigma$ of $\psi n$ numbers and put an edge between the vertices $v_i$ and $u_{\sigma(i)}$ for each $1\leq i\leq \psi n$. 
\end{defin}

With the following lemma we bound the mixing time of $\cH_{n,d,\psi}$.
\begin{lemma}\label{Hmixing} 
In the graph  $\cH_{n,d,\psi}$; where  $ 0<\psi<1$, and $d$ or $\psi$ can be a function of $n$, if $d=\omega(\log^2 (n))$, we have: $\tmix(\cH_{n,d,\psi})=\bigO(\frac{\log n}{\log d}+ \frac{d}{\psi})$, and $\davg={\Theta}(d)$. 
%If $d=o(\log^2(n))$, we have: $\tmix(\cH_{n,d,\psi})=\bigO({\log^2 n}+ \frac{d}{\psi})$, and $\davg={\Theta}(d)$. 
\end{lemma}

\begin{proof}
 Once again, we use Theorem \ref{SP} by Peres and Sousi.  We will show that $t^{H}_{\cH_{n,d,\psi}}(1/4)=\bigO(\frac{\log n}{\log d}+ \frac{d}{\psi} )$. Consider a set $S$ with $\pi(S)\geq1/4$, and assume the heavier part of $S$ is in $H_1$. Consider arbitrary $w\in H_2$, we have $\E(\tau_w(S))\leq E(\tau_w(S\cap H_1))\leq \E(\tau_w(H_1))+\max_{z\in H_1}\E(\tau_{z}(S\cap H_1)) $. By  a result of Hildebrand (\cite{Hilder}), we know if $d=\omega(\log^2 n)$ then $\tmix(H_1)=\tmix(H_2)= \bigO(\log n/\log d)$.  Thus, employing Theorem \ref{SP}, we know that $ \max_{z\in H_1} \E(\tau_{z}(S\cap H_1))\leq c_{1/8} \frac{\log n}{\log d}$. Note that the number of steps needed to hit $H_1$ from a $v\in H_2$ is a geometric random variable with probability $\psi/d$ of success, thus $\E(\tau_w(H_1))=d/\psi$. Therefore, $\E(\tau_w(S))\leq c_{1/8} \frac{\log n}{\log d}+\frac{d}{\psi}$. Since $w\in H_1$ and the heavier part of $S$ in $H_2$ is the worst case, we can conclude the result. 

%For $o(\log^2 (n))$, note that by a result of Benjamini et al. \cite{BenjRandomGraph} we know for an \ER graph with parameters $\langle n, p=c/n\rangle$; $c>1$ being a constant we will have a giant connected component. i.e one containing a constant fraction of vertices, and the mixing time on this giant component is $\bigO(\log^2 n)$.  The rest of the prove will be very similar to the case of $d=\omega(\log^2(n))$.

The average degree for each vertex is $d+\psi$, since $0<\psi<1$, it is obviously $\Theta(d)$.
\end{proof}

%In the rest of this text we take $\G$ to be the graph obtained from $\cH_{n,d,\psi}$ after employing the decoration  construction. When we talk about an algorithm we mean one with the characteristics which were mentioned in the first paragraph of this article. 

In the rest of this Section, we use the following terminology: 
\begin{defin}When a vertex is queried by an algorithm we say it is $\quer$ or it belongs to the set $\quer$. By saying a vertex is $\obs$, or that it belongs to the set $\obs$, we mean that at least two of its neighbors have been $\quer$.  \end{defin}
We write $\obs$ because, if the algorithm queries two neighbors of a vertex $v$, then the algorithm can obtain information about $v$ (without querying $v$)  --- indeed, if the two neighbors have degree more than $1$  then necessarily $v$ is not $\stared$. Thus, if a vertex is $\obs$ or $\quer$ we will be able to decide with certainty whether or not it is $\stared$.  On the other hand,  if a vertex is not $\obs$, nor $\quer$ there is no information available about it.
 
 \smallskip
 
We now show that, with large enough probability, if few queries are performed then no vertex will be $\obs$, and the number of edges that will be revealed will be relatively small.
 \begin{lemma}\label{tree}
  Let $\G$  be the graph obtained from $\cH_{n,d,\psi}$ after employing the decoration  construction.
Consider an algorithm that traverses the edges of $\G$ and queries $q$ vertices of it. The probability of having at least one $\obs$ vertex which has not been $\quer$ is at most  ${{q}\choose{2}} \frac{(d+\psi)^2}{n} $.
 Furthermore, the $\quer$ vertices will constitute an induced graph  in $\G$ which in expectation has  $E_{\quer}+(q-1)$ edges where $E_{\quer}$ is a random variable with expected value $q^2\left(\frac{d}{n} \right)$, and   thus, $\cP\left(E_{\quer}\geq 100 q^2\left(\frac{d}{n} \right)\right)\leq 1/100$. \end{lemma}

\begin{proof}
Any $\obs$ vertex is connected to at least two distinct $\quer$ vertices. Thus, the probability of having an $\obs$ vertex is equal to the probability that at least two distinct $\quer$ vertices share one neighboring vertex. i.e. $\exists v,u\in \quer; N_u\cap N_v\neq \emptyset$.
In  $\cH_{n,d,\psi}$, if two vertices belong to the same  \ER graph the probability of an edge being present between them is $d/n$ and if the vertices belong to two different \ER graphs the probability of having an edge between them is $\psi /n$. Thus, the probability of having an edge is  $\max \{\frac{d}{n},\frac{\psi}{ n}\}$. Thus, the probability that two specific vertices share at least one neighbor is less than: $n (\frac{d}{n}+\frac{\psi}{n})^2 =(d+\psi)^2/n$.
By the union bound, among the $q$ $\quer$ vertices, the probability of having at least two vertices sharing a neighbor is less than or equal to:  ${{q}\choose{2}} \frac{(d+\psi)^2}{n} $. 

\smallskip

For the second part 
assume we have queried $q_1$ vertices from $H_1$ and $q_2$ vertices in $H_2$. The expected number of edges that we have not traversed in our walk but exist between $\quer$ vertices
 is: ${{q_1}\choose{2}}\frac{d}{n}+{{q_2}\choose{2}}\frac{d}{n} + q_1q_2\left( \frac{\psi}{n}\right)$.  
 
 Thus, if $\#E_{\quer}$ is the number of  edges between the $\quer$ vertices that have not been traversed, We have $\E(\#E_{\quer})\leq q^2\left( \frac{d}{n} +  \frac{\psi}{2n} \right) $ using the Markov inequality, we will get the result. 

 %Chernoff bound,  we have $ \cP(\#E_{\quer}>2\left( {{q_1}\choose{2}}\frac{d}{n}+{{q_2}\choose{2}}\frac{d}{n} + q_1q_2\left( \frac{\psi}{n}\right)\right)\leq e^{-{{q_1}\choose{2}}\frac{d}{n}+{{q_2}\choose{2}}\frac{d}{n} + q_1q_2\left( \frac{\psi}{n}\right)/3}$. Thus, $ \cP(\#E_{\quer}>2\left( \frac{q^2}{2} \frac{d}{n} + \frac{q^2}{4} \frac{\psi}{n} \right)\leq e^{-\left(\frac{q^2}{8}\right)\left(\frac{d}{n}\right) /3}$.
 
\end{proof}

Finally, we prove the following lemma, that will be the heart of the proofs of  Theorems \ref{main21},  \ref{main22}, and \ref{main23}:
%%%%%%%%MAIN LEMMA%%%%%%%
 \begin{lemma}\label{finalLemma}
Consider arbitrary $n$, $d>\omega(\log n)$, $\Omega(\log n)<t<o({{n}}/{d^2})$, so that $t/d=\Omega(1)$,
%If $t=\bigO(\log n/\log d)$
%take $\G$ to be the graph obtained from the decoration construction applied to a \ER graph with parameter $n,d$. Otherwise, 
take $\G$ to be the graph obtained from the decoration construction applied to  $\cH_{n,d,d/t}$\footnote{
If $c > d/t\geq 1$, for constant $c$, let $\psi = d/tc$. }. By Lemma \ref{Hmixing}, we know that $\tmix (\G)=\Theta(t)$ and $\davg=\Theta(d)$.

\smallskip

If, instead, $t=O(\log n/\log d)$, and $d=\Theta(\log^d n)$, take $\G$ to be the decorated version of an \ER graph with parameters $\langle n, d/n \rangle$.

\smallskip

Then,
\begin{itemize}

\item if an algorithm traverses the edges of $\G$ and queries at most  $t d/c$ vertices of $\G$; $c$ being a constant, then with probability at least $99/100- 202/(c-1)$ there is no $\stared$ vertex among its $\quer$ vertices.
\item  If an algorithm traverses the edges of $\G$ and queries $q\leq \frac{n}{cd}$ vertices of $\G$; $c$ being a constant, then with probability  $1- o(1)$ the expected number of $\stared$  centers which have been queried is less than $ \frac{8c}{c-1}\left(\frac{q}{dt}\right)$.
\end{itemize}
 \end{lemma}

\begin{proof}
 Consider all the $\quer$ vertices, and assume they are queried in the following order: 
 $v_1,v_2,\dots, v_q$, where $q\leq dt/c$. 
 %By Lemma \ref{tree} we know that with probability  at least  $1-\bigO\left(q^2 \frac{d^2}{n}\right)=1-o(1) $, these vertices constitute a tree.  
 For each $i$, let $d_i$ be the number of  $\quer$ neighbors of $v_i$  i.e. $d_i:=\vert N_{v_i}\cap \quer\vert$. By Lemma \ref{tree}, with probability  $99/100$ we have, $\sum_{i=1}^{q}d_i\leq2\left(100q^2\left(\frac{d}{n} +  \frac{\psi}{2n} \right)+ (q-1)\right)$. Note that having $q=o(n/d)$, we will have: $\sum_{i=1}^{q}d_i\leq2( 101q-1)$. For each $v_i$ let $\sigma(v_i)$ be the vertex $v_j$ with minimum $j$ which is adjacent to $v_i$.  
 
 From now on we will abuse the notation and by saying a $\stared$ vertex we mean only a $\stared$ center. Note that to query any $\stared$ vertex we need to query the center  first thus not querying any center is equivalent to not querying any $\stared$ vertex. To query a $\stared$ vertex we need to query its $\marked$ neighbor. Thus, $\cP(v_i  \in \stared)= \cP(\sigma(v_i) \in   \marked)\cP(v_i\in  \stared\vert\sigma(v_i)   \in\marked).$ 
 Let   $S_k$ be the set of $v_i$s which have been $\quer$ and satisfy $\sigma(v_i)=v_k$. 
  Note that $\cP(\exists v\in S_k \cap  \stared\vert\sigma(v_k)   \in\marked)$ is at most  $d_{v_k}/\deff(v_k)$ where $\deff(v_k)$ is the number of $v_k$'s neighbors which have not been $\obs$.

 We will have: 

  $$
 \begin{array}{l}
\hspace{0.94cm} \E(\# \text{{ \quer  }}\cap \stared )=\\
 \hspace{3cm}= \sum_{i=1}^{q}  \cP(\text{\tiny Having a  } \stared \text{\tiny vertex in }S_i )\\
  \hspace{3cm}=\sum_{i=1}^{q}\cP(\text{\tiny Having a  } \stared \text{\tiny vertex in }S_i\vert v_i \space\text{ \tiny is }\marked)\cP(v_i \text{ \tiny is } \marked)\\
   \hspace{3cm}\leq\sum_{i=1}^{q}(\frac{ d_i}{\deff(v_i)})(\frac{ 1}{t}).
  \end{array}
 $$

Note that for any $i$, $\E(\deff(v_i))=d(1-\frac{qd}{n})$, thus by Chernoff bound the probability that a single vertex in an \ER graph has effective degree less than $1/2d(1-qd/n)$ is at most $e^{-d(1-qd/n)/8}$. Considering all the $q$ vertices this probability will be $qe^{-d(1-qd/n)/8}$, plugging in the values for $q$ and $d$ we will have: 
$$
\begin{array}{ll}
\cP\left(\text{\tiny any vertex among the }\quer \text{ \tiny has degree less than }\frac{d}{2}(1-\frac{qd}{n})\right)\leq q e^{-\frac{d}{8}(1-\frac{qd}{n})}\\
\hspace{9.43cm}\leq q \left(e^{-\frac{d}{8}(1 -o(1))}\right)\\
\hspace{9.43cm}= o(1).\\ \text{\hspace{6cm} given } d=\omega(\log n ) \text{, and } q=o(\frac{n}{d}).
\end{array}$$

Hence, with probability $1-  o(1)$, if $q=td/c$ we have:

  $$
 \begin{array}{l}
 \cP(\text{\tiny{ Having a  }} \stared \text{\tiny vertex} )=\\
  \cP(\text{\tiny{ Having a  }} \stared \text{\tiny vertex} \vert \sum_{i=1}^{q} d_i\leq 202q)\cP(\sum_{i=1}^{q} d_i\leq 202q)+\cP(\sum_{i=1}^{q} d_i\geq202q)=\\
 
   \hspace{5cm}\leq\sum_{i=1}^{q}(\frac{ d_i}{\deff(v_i)})(\frac{ 1}{t})(\frac{99}{100})+\frac{1}{100}\\
      \hspace{5cm}\leq\sum_{i=1}^{q}(\frac{ d_i}{d(c-1)/2c})(\frac{ 1}{t})(\frac{99}{100})+\frac{1}{100}\\

            \hspace{5cm}\leq (\frac{99}{100})\frac{2c}{dt(c-1)}\sum_{i=1}^{q} d_i+\frac{1}{100}\\
       
                 \hspace{5cm}     \leq (\frac{99}{100})\frac{202qc}{dt(c-1)}+\frac{1}{100}\leq \frac{1}{100}+\frac{202}{c-1}.
  \end{array}
 $$
 
 If $q$ is arbitrary and yet less than $o(n/d)$ by repeating the above calculations, and  employment of Lemma \ref{tree} we get: 
 
  $$
 \begin{array}{l}
 \E(\# { \quer  }\cap \stared  )\leq
                   \frac{2}{td(1-q(d/n))}\E(\sum_{i=1}^{q} d_i)
                    \leq \frac{8q}{dt(1-o(1))} . \end{array}
 $$

\end{proof}
 
We can finally prove our three main Theorems:\begin{itemize}
\item  \emph{Proof of Theorem \ref{main21}.}
 Consider the two graphs $\G_1$ and $\G_2$, $\G_1$ being the graph of Lemma \ref{finalLemma} with $c_1=1$ and 
 $\G_2$ the same graph without the $\stared$ vertices (the graph before the decoration construction). 
%<<<<<<< .mine
Any algorithm which queries less than $\tmix\davg/c$ vertices of $\G_1$ or $\G_2$ will fail to distinguish between them with probability at least $\frac{1}{100}+\frac{202}{c-1}$.
Let $\Pi^0_{\G_i}$ be the uniform distribution on vertices of $\G_i$; $i=1,2$. In $\G_1$ with high probability we have at least $2\vert V_{\G_1}\vert $ $\stared$ vertices. 

 \begin{itemize}
\item \emph{Part 1.} Note that $\vert\Pi^0_{\G_1},\Pi^0_{\G_2}\vert_1\geq 1/4$. 
 Thus, if the natures selects $\G_1$, any algorithm $\cal A$ which aims to outputs $\Pi^0_{\G_1}$ will return a distribution $\Pi_{\cal A}$ satisfying 
$\vert\Pi^0_{\G_1},\Pi_{\cal A}\vert_1\geq 1/4-\frac{1}{100}-\frac{202}{c-1}$.

 \item\emph{ Part 2.}
  For $\zeta>1$, let $\Pi^\zeta_{\G_i}$ be the probability  distribution on vertices of $\G_i$; $i=1,2$ weighing each vertex $v$ proportional to $deg(v)^{\zeta}$.  If we take a sample from distribution $\Pi^\zeta_{\G_1}$ it will be $\stared$ vertex with probability $(t^{\zeta-1} +1)/2(t^{\zeta-1} +d^\zeta+1)$. Thus, for  $t^{\zeta-1} \geq d^\zeta$ we have, 
  $\vert\Pi^{\zeta}_{\G_1},\Pi^{\zeta}_{\cal A}\vert_1\geq 1/4$. Thus, 
  $\vert\Pi^0_{\G_1},\Pi_{\cal A}\vert_1\geq 24/100-202/c$.\hspace{0.5cm}$\Box{}$
\end{itemize}
\smallskip
 
\item \emph{Proof of Theorem \ref{main23}.}
%>>>>>>> .r132
Take the two graphs $\G_1$ and $\G_2$, $\G_1$ being the graph of Lemma \ref{finalLemma}, and  $\G_2$ the same graph without the $\stared$ vertices (the graph before the decoration construction). We have: $\E(\vert V_{\G_2}\vert)=n$, $\E(\vert V_{\G_1}\vert )=(1+c_1)n$, and
 $\E(\davg({\G_2})\vert)=d$, $\E(\vert \davg({\G_1})\vert )= \frac{d+c_1}{1+c_1}$.\hspace{0.5cm}$\Box{}$
 \item \emph{Proof of Theorem \ref{main22}.}
Take $c_1=1$, and let $\G=\langle V_{\G},E_{\G}\rangle$ be the graph constructed as in Lemma \ref{finalLemma}. 
Consider two functions $\calF_1:V_{\G}\rightarrow [0,1]$ and $\calF_2:V_{\G}\rightarrow [0,1]$. Let the function $\forall v\in V_{\G}\setminus \stared , \calF_1(v)=\calF_2(v)=0$. For any $v\in \stared$ we set $\calF_1(v)=1$ with probability $1/2+\epsilon$ and $\calF_2(v)=1$ with probability $1/2-\epsilon$. 

 Note that $\vert \calF_1-\calF_2\vert_1\geq \epsilon$, and by employing the following classical result \cite{SamplingLowerBound},
with probability $1-o(1)$ we will not be able to distinguish between $\calF_1$ and $\calF_2$.
\begin{lemma}[\cite{SamplingLowerBound}] Consider a $\left(\frac12-\epsilon,\frac12+\epsilon\right)$-biased coin (that is, a coin whose most likely outcome has probability $\frac12+\epsilon$. To determine with probability at least $1-\delta$ what is the most likely outcome of the coin, one needs at least  $\Omega(1/\epsilon^2 \log(1/\delta))$ coin flips.\end{lemma}  

%Assume that we have queried less than $\tmix\davg\log(1/\delta)(1/\epsilon^2)$ vertices. 

By Lemma \ref{finalLemma}  with probability $1-o(1)$,  the expected number of $\stared$ centers will be $\frac{8q}{dt(1-o(1))}$. Thus, if $q=\omega(dt)$

%Plugging in $q=(\tmix\davg)\log(1/\delta)(1/\epsilon^2)/16$, and employing the Chernoff bound we get: 
$$ \cP\left(\#  \text{{ \quer }}\cap \stared  \text{\tiny{ center }}\geq 48 q/ dt\right) \leq  e^{-16 q/ dt}\leq o(1). $$

Therefore, since in order   to distinguish between $\calF_1$ and $\calF_2$ with probability at least $1-\delta$, we need at to see at least $\Omega\left(\log(1/\delta)(1/\epsilon^2)\right)$ starred centers, or equivalently $\Omega\left(dt\log(1/\delta)(1/\epsilon^2)\right)$ queries.\hspace{0.5cm}$\Box{}$
 \end{itemize}
 \section{Conclusion}
 In this paper we have studied the complexity of computing a number of functions of online graphs, such as online social networks, in terms of their average degree and their mixing time. We have obtained a tight bound for the problem of computing the average of a bounded function on the vertices of the graph (e.g., the average star rating of a movie), and a near-tight bound for the problem of sampling a close-to uniform-at-random vertex (many algorithms in the literature assume to have access to such an oracle), and lower bounds for the problems of estimating the order, and the average degree of the graphs.
 
 \smallskip
 
It will be interesting to pursue the study of  these online graphs problems in order to bridge the gap between theoretical algorithms, and applied ones. Besides the obvious questions (what are the optimal bounds  for estimating the order and the average degree of a graph?), an interesting open problem is to understand which  structural properties of online social networks could be used by algorithms to improve the complexity of the various problems that practitioners have been considering.

\section*{Acknowledgments}
The authors would like to thank Anirban Dasgupta, Ravi Kumar, Silvio Lattanzi and Tam\'as Sarl\'os for several useful discussions.

 \newpage
 
\appendix

\section{Appendix}

\subsection{Proof of Theorem~\ref{main12}}\label{app:mdub}
We employ the following Theorem due to Chung et al.\cite{ChernofMark}:
 
 \begin{theorem}\label{thmChernoffMarkof}
 Let $ \cal M$ be an ergodic Markov chain on sample space $S$ and $s_1,s_2,\dots,s_t$, $t$ steps of this chain starting from distribution $\phi$. Let $\pi$ be the stationary distribution of $\cal M$, $f:\Omega\rightarrow [0,1]$, and $X=\sum_{i=1}^{t}f(s_i)$.
 For $0\leq\delta\leq1$ we have:
 \begin{equation}\label{Eq3}
 \begin{array}{ll}
 \cP \left(  \vert X-\mu t \vert \geq \delta \mu t\right) \leq C \parallel \phi \parallel_\pi e^{-\frac{\delta^2\mu t}{72 \tmix } }\\ \text{Or for}~ t>c\tmix,\\
  \cP \left(  \vert X/t-\mu  \vert \geq \delta \right) \leq C \parallel \phi \parallel_\pi e^{-\frac{\delta^2 c}{72 } }. 
  \end{array}
 \end{equation}
where $\mu=\sum_{v\in \Omega} f(v)\pi(v)$, and $\tmix$ is the mixing time with $1/8$ 
\footnote{Note that to get $\frac{1}{8}$ close to the stationarity $\Theta(\tmix)$ steps suffices.}
precision.
 \end{theorem}
 
% \begin{coro}\label{Cor1}
% For any $v\in \Omega$ let $\counter_{\cal M}(v)$ be the number of times the Markov chain $\cal M$ with mixing time $\tmix$ visits $v$. For any $ct>\tmix$,  we have:  
% \begin{equation}
% \cP \left(  \vert  \counter_{\cal M}(v) -\pi(v)t \vert \geq \delta \right) \leq C \parallel \phi \parallel_\pi e^{-\frac{\delta^2 c}{72 } }.
% \end{equation}
% \end{coro}
 
 \vspace{0.3 cm}
 
 \emph{Proof of Theorem \ref{main21}}.
 Note that 
after $\tmix \D/d_{min}$ steps we will start the chain from distribution $\phi$ satisfying $\parallel \phi\parallel_{\pi}\leq 1/4$. Then employing Equation \ref{Eq3} for $t=\tmix\D/d_{min}1/\epsilon^2\log(1/\delta)$, and knowing by  Chierichetti et al. result(\cite{FlavioKumar}) that for the Markov chain $\cal M$  the stationary distribution is the uniform distribution  and mixing time is $\tmix\D/d_{min}$,
we will have:
 $$\cP \left(  \vert X/t-\mu \vert \geq \epsilon \right) \leq \frac{C}{4}  e^{-\frac{\epsilon^2\log(1/\delta)}{72 \epsilon^2 }}\leq \delta.$$

 Note that in expectation we stay at vertex $v$ for $\D/\deg(v)$ steps. Thus, using Markov inequality with high probability after $\tmix\frac{\davg}{d_{min}}1/\epsilon^2\log(1/\delta)$ queries $\cal M$ has taken $\tmix\frac{\D}{d_{min}}1/\epsilon^2\log(1/\delta)$ steps. Equivalently, considering the variables $t$ and $i$ in the pseudocode, when  $t>\tmix\frac{\davg}{d_{min}}\frac{1}{\epsilon^2}\log(\frac{1}{\delta})$, with high probability  $i>\tmix\frac{\D}{d_{min}}({1}/{\epsilon^2})\log(1/\delta)$. \hspace{8.3cm}
$\Box{}$

\subsection{Other Algorithms}\label{sec:otheralgos}

\begin{algorithm}
\begin{algorithmic}[1]
\REQUIRE Seed vertex $s\in V_{\G}$, a constant approximation of $\tmix$
\ENSURE An $\epsilon$ additive approximation of $f_{avg}$ with probability at least $1-\delta$
\STATE{Consider the simple random walk on $\G$}
\STATE{$T\leftarrow\tmix \ln \epsilon^{-1}$}
\STATE{Starting from $s$, run the chain for $T \cdot  (\davg \epsilon^{-2} \ln \delta^{-1})$ steps (in order to produce $\davg \epsilon^{-2} \ln \delta^{-1}$ close-to-independent samples) --- let $v_0 = s, v_1, v_2, \ldots$ be the states that are visited by the walk}
\STATE{$S\leftarrow0$}
 \FOR{$i=1, 2, \dots , \davg \epsilon^{-2} \ln \delta^{-1}$}
 \STATE{Flip a coin with heads probability $1/\deg(v_{T \cdot i})$}
 \STATE{If the coin comes up heads, $S\leftarrow S+{\cal F}(v_{T \cdot i})$}
\ENDFOR
\RETURN{$S /(\davg/\epsilon^2 \log(1/\delta))$}
\end{algorithmic}
\caption{\label{alg:rej}The Rejection Algorithm.}
\end{algorithm}

\textbf{The  rejection  algorithm queries $\mathbf{ \tmix \davg\log(\epsilon^{-1}) \epsilon^{-2}\log(\delta^{-1})}$ vertices.}
The proof is an application of the Chernoff bound and the result from \cite{FlavioKumar}. Note that here since we need all the samples to be taken with probability $\epsilon$ close to their probability at stationarity, we are not able to employ Theorem \ref{thmChernoffMarkof}.

\begin{algorithm}
\begin{algorithmic}[1]
\REQUIRE Seed vertex $s\in V_{\G}$, a constant approximation of $\tmix$, and a $(1\pm \epsilon)$ approximation of $\davg$
\ENSURE An $\epsilon$ additive approximation of $f_{avg}$ with probability at least $1-\delta$
\STATE{Consider the simple random walk on $\G$}
\STATE{Starting from $s$, run the chain for $\tmix \left(\ln\epsilon^{-1} + \davgs \epsilon^{-2} \ln\delta^{-1}\right)$ steps --- let $v_0 = s, v_1, v_2, \ldots$ be the states that are visited by the walk}
\STATE{$T \leftarrow \tmix\davgs \epsilon^{-2} \ln(1/\delta)$}
\FOR{$i=1, 2, \dots , T$}
\STATE{$S\leftarrow S+ \left({\cal F}(v_{i + \tmix \ln \epsilon^{-1}})/\deg(v_{i + \tmix \ln \epsilon^{-1}})\right)$}
\ENDFOR
\RETURN{$S \cdot \davg/T$}
\end{algorithmic}
\caption{\label{alg:weight}The Weighted Sampling Algorithm.}
\end{algorithm}

 \vspace{2 mm}
 
 \textbf{The weighted sampling algorithm queries $\mathbf{ \tmix{\cal D}^2 \davg \epsilon^{-2}\log(\delta^{-1})}$ vertices.}
Take $f_1:V\rightarrow [0,1]$ be $f_1(v)=\calF(v)/\deg(v)$. 
 Note that 
 $$
 \begin{array}{ll}
 \E(f_1)&=\sum_{v\in V} {\cal F}(v)\pi(v)/\deg(v)\\
 &= \sum_{v\in V}{\cal F}(v)/|E|\\
 &=(n/\vert E\vert) f_{avg}.
\end{array}$$
 We need to take $\epsilon'=\epsilon/\davg$, the rest follows from application of Theorem \ref{thmChernoffMarkof} to find an $\epsilon'$ approximation to $\E(f_1)$ and then multiply it by $\davg$.\hspace{23mm}
\medskip

\textbf{The Metropolis algorithm suggested in the literature queries $\mathbf{ \tmix{\cal D}^2 \davg \epsilon^{-2}\log(\delta^{-1})}$ vertices.} The analysis will be similar to the analysis of previous algorithms. Note that the Metropolis algorithm is not even an optimal method for taking uniform samples (Problem \ref{pr1}). 


\begin{thebibliography}{10}


\bibitem {HitMix}Y. Peres, P. Sousi, \emph{ Mixing times are hitting times of large sets.}  Journal of Theoretical Probability, Vol. 28, Issue  2, Pages 488--519, 2015.

\bibitem{FlavioKumar} 
F. Chierichetti, A. Dasgupta, R. Kumar, S. Lattanzi, T. Sarl$\rm{\acute{a}}$s,\emph{ On Sampling Nodes in a Network.}  Proceedings of the 25th International Conference on World Wide Web, Pages 471--481, 2016. 

\bibitem{Hilder} M. Hildebrand,\emph{ Random walks on random simple graphs.}
Random Structures and Algorithms, Pages 301--318, 1996.

\bibitem{FReed} N. Fountoulakis, B. A. Reed \emph{The evolution of the mixing rate of a
simple random walk on the giant component of a random graph.} Random Structures
Algorithms, Vol. 33, Pages 68--86, 2008. 

\bibitem{ChernofMark} K. M. Chung, H. Lam, Z. Liu, and M. Mitzenmacher, \emph{Chernoff-Hoeffding bounds for Markov chains:  generalized and simplified.}  29th Symposium on Theoretical Aspects of Computer Science (STACS), Pages 124--135, 2012. 

\bibitem{SamplingLowerBound} R. Canetti, G. Even, O. Goldreich, \emph{Lower Bounds for Sampling Algorithms for Estimating the Average.} Journal of
Information Processing Letters, 
Vol. 53, Issue 1, Pages 17--25, 1995. 
 
 
 \bibitem{CooperParams}    C. Cooper, T. Radzik, Y. Siantos, \emph{Estimating Network Parameters Using Random Walks.} 4th International Conference on Computational Aspects of Social Networks (CASoN), Pages 33--40, 2012.




\bibitem{KumarAvg}  A. Dasgupta, R. Kumar, T. Sarlos,	
\emph{On estimating the Average Degree.} Proceedings of the 23rd international conference on World wide web (WWW),
Pages 795--806, 2014. 


\bibitem{AvgFeige} U. Feige, \emph{On Sums of Independent Random Variables with Unbounded Variance and Estimating the Average Degree in a Graph.} SIAM Journal on Computing, Vol .35, Num. 4, Pages 964--984, 2006. 

\bibitem{AvgGoldreich} O. Goldreich, D. Ron, \emph{Approximating Average Parameters of Graphs.} 
Journal of Random Structures and Algorithms,
Vol. 32, Issue 4, Pages 473--493, 2008. 


\bibitem{Katzir} L. Katzir, E. Liberty, O. Somekh, \emph{
Estimating Sizes of Social Networks via Biased Sampling.} 
 Proceedings of the 20th international conference on World wide web (WWW), Hyderabad, India, Vol. 10, Pages 335--359, 2011. 




\bibitem{LowSize} 
V. Kanade, F. Mallmann-Trenn, V.Verdugo, \emph{How large is your graph?} 
\arx{1702.03959}

 \bibitem{Dasgupta12} {A. Dasgupta, R. Kumar, and D. Sivakumar},
 \emph {Social Sampling.}
  {Proceedings of the 18th ACM SIGKDD International Conference on Knowledge Discovery and Data Mining},
{KDD '12},
 {Beijing, China},
 Pages  {235--243},  {2012}.
 
 
 \bibitem{Ahn07}
{Y, Ahn, S.  Han, H. Kwak, S. Moon, and H. Jeong},
\emph{Analysis of Topological Characteristics of Huge Online Social Networking Services.}
{Proceedings of the 16th International Conference on World Wide Web},
{WWW '07},
 Pages {835--844}, 2007.

 
\bibitem{Mislove07}
  {A. Mislove, M. Marcon, K. P. Gummadi, Krishna P. , P. Druschel, and B. Bhattacharjee},
 \emph{Measurement and Analysis of Online Social Networks.}
 {Proceedings of the 7th ACM SIGCOMM Conference on Internet Measurement},
{IMC '07},
 {San Diego, California, USA},
Pages {29--42}, {2007}.


\bibitem{Becchetti:2010}
    L. Becchetti, P.  Boldi, C. Castillo, and A. Gionis,
    \emph{Efficient Algorithms for Large-scale Local Triangle Counting.}
    {ACM Trans. Knowl. Discov. Data},
 Pages  {13:1--13:28}, {2010}.
\bibitem{Bressan}
{M. Bressan, E. Peserico, and L. Pretto},
      \emph{Simple set cardinality estimation through random sampling.}
      {arXiv:1512.07901}, {2015}.
\bibitem{numedges}
    {T. Eden and W. Rosenbaum},
\emph{On Sampling Edges Almost Uniformly}.
 {arXiv:1706.09748}, {2017}.


\bibitem{LKF05}
       J. Leskovec, J. Kleinberg,  and C. Faloutsos,
           \emph{        Graphs over Time: Densification Laws, Shrinking Diameters and Possible Explanations},
        {Proceedings of the Eleventh ACM SIGKDD International Conference on Knowledge Discovery in Data Mining}, {KDD '05},
{Chicago, Illinois, USA},
        Pages {177--187}, 2005.
  
\bibitem{BenjRandomGraph}
{I. Benjamini, G. Kozma,  and N. Wormald},
   \emph{The Mixing Time of the Giant Component of a Random Graph},
        {Random Structures \& Algorithms},
              Vol. {45},  Pages  {383--407},  {2014}.

\bibitem{Leskovec09}
       {J. Leskovec, K. J. Lang, Kevin, A.  Dasgupta, Anirban and M. W. Mahoney},
       \emph{ Community Structure in Large Networks: Natural Cluster Sizes and the Absence of Large Well-Defined Clusters.}
Internet Mathematics,
  Internet Math.,
        Num. {1}, Vol. 6,
        Pages {29--123},
        Vol. {6},
   {2009}.


\bibitem{Seshadhri14}
        Acmid = {2913420},
{C. Seshadhri, A. Pinar, Ali and T. G. Kolda},
  \emph{Wedge Sampling for Computing Clustering Coefficients and Triangle Counts on Large Graphs.} {Statistical Analysis and Data Mining},
        Num. {4}, Vol. 7, 
        Pages {294--307},
 {2014}.

\bibitem{Gjoka:2010}
      M. Gjoka, M. Kurant, C. T. Butts,  and A. Markopoulou, 
      \emph{Walking in Facebook: A Case Study of Unbiased Sampling of OSNs.}
    {Proceedings of the 29th Conference on Information Communications},
       {San Diego, California, USA},
        Pages {2498--2506},
        {2010}.
  

\bibitem{Clustering}
        {S. T. Schank and W. Dorothea},
        \emph{Approximating Clustering Coefficient and Transitivity.}
{Journal of Graph Algorithms and Applications},
     Vol. 9.    Pages  {265--275},
{2005}.

\bibitem{Cycle}
N. Alon,
R.  Yuster,
and U.  Zwick,
\emph{Finding and counting given length cycles.}
Algorithmica ,
Vol. 17,
Num. 3,
Pages 209--223,
1997.


\bibitem{Hardiman}
S.J. Hardiman, P. Richmond,  and S. Hutzler,
\emph{Calculating statistics of complex networks through random walks with an application to the on-line social network Bebo.}
The European Physical Journal B,
Vol. {71},
Num. {4},
Pages {611},
{2009}.




\end{thebibliography}
\end{document}